\documentclass[letterpaper]{article}
\usepackage{natbib}
\usepackage[margin=1.05in]{geometry} 
\usepackage[utf8]{inputenc}
\usepackage{amsmath,amssymb,amsfonts,amsthm} 
\usepackage{enumerate}
\usepackage{pgf, tikz}
\usepackage{subfig} 
\usetikzlibrary{positioning}
\usepackage{hyperref}
\usepackage[capitalize]{cleveref}
\usetikzlibrary{arrows, automata}
\usepackage{algpseudocode}
\usepackage{algorithm}
\usepackage{bbm}
\usepackage{nth}
\algnewcommand\algorithmicforeach{\textbf{for each:}}
\algnewcommand\ForEach{\item[ \algorithmicforeach]}

    \newtheorem{theorem}{Theorem}
    \newtheorem*{proposition*}{Proposition}
    \newtheorem{lemma}{Lemma}

    \newtheorem{corollary}{Corollary}
    
    \newtheorem{observation}{Observation}
     \theoremstyle{definition} 

\def\x*{\mathbf{x^*}}
\def\x{\mathbf{x}}
\def\v{\mathbf{v}}
\def\px{\mathbf{\pi}}
\def\px*{\mathbf{\pi^*}}

\def\q{\mathbf{q}}
\def\w{\mathbf{w}}
\def\y{\mathbf{y}}
\def\z{\mathbf{z}}

\def\c{\mathbf{c}}

\def\0{\mathbf{0}}
\def\1{\mathbf{1}}

   \usepackage{xcolor}

\title{The Strong Maximum Circulation Algorithm: A New Method for Aggregating Preference Rankings}

\author{Nathan Atkinson \\ University of Wisconsin \\ \tt{natkinson@wisc.edu}  \and Scott C. Ganz \\  Georgetown MSB and AEI \\ \tt{scott.ganz@georgetown.edu} \and Dorit S. Hochbaum \\ UC Berkeley \\ \tt{dhochbaum@berkeley.edu} \and James B. Orlin \\ MIT Sloan \\ \tt{jorlin@mit.edu}}

\begin{document}
\date{}
\maketitle

\begin{abstract}
We present a new optimization-based method for aggregating preferences in settings where each voter expresses preferences over pairs of alternatives. Our approach to identifying a consensus partial order is motivated by the observation that collections of votes that form a cycle can be treated as collective ties. Our approach then removes unions of cycles of votes, or circulations, from the vote graph and determines aggregate preferences from the remainder. Specifically, we study the removal of \emph{maximal circulations} attained by any union of cycles the removal of which leaves an acyclic graph. We introduce the \textit{strong maximum circulation}, the removal of which guarantees a unique outcome in terms of the induced partial order, called the \textit{strong partial order}. The strong maximum circulation also satisfies strong complementary slackness conditions, and is shown to be solved efficiently as a network flow problem. We further establish the relationship between the dual of the maximum circulation problem and Kemeny's method, a popular optimization-based approach for preference aggregation. We also show that identifying a {\em minimum maximal} circulation---i.e., a maximal circulation containing the smallest number of votes---is an NP-hard problem. Further an instance of the minimum maximal circulation may have multiple optimal solutions whose removal results in conflicting partial orders.\\

\end{abstract}

\newpage

\section{Introduction}

The problem of aggregating preferences is of central interest to a broad set of problems in economics, political science, psychology, operations research, and computer science. Given a set of (possibly conflicting) pairwise comparisons of alternatives, how can we determine non-conflicting aggregate pairwise comparisons that best summarize the data? This question is fundamental to topics as diverse as the analysis of democratic voting systems, the design of online search and recommendation algorithms, and the construction of nonparametric statistical tests.

It is well-known that no preference aggregation method satisfies every potential desirable normative property---e.g., non-dictatorship, independence of irrelevant alternatives, and Pareto efficiency \citep{arrow_social_1951,pini2009aggregating}. Further, existing approaches that are based on maximizing the alignment of an aggregate ranking with a set of pairwise comparisons face issues related to computational complexity that are driven by the exponential size of the set of possible orderings of alternatives \citep{bartholdi_voting_1989, fischer2016weighted}. As a result, there has been significant interest in applying approaches from optimization and machine learning in order to discover new aggregation methods that are both well-justified normatively and computationally efficient \citep[see, e.g.,][]{agarwal_mathematics_2010}.

We introduce a new optimization-based approach for aggregating preferences that is based on the logic of removing cycles of votes, i.e., sets of votes for which every vote in favor of an alternative is counterbalanced by a vote against that alternative. Given a cyclic set of votes, there is no principled way of saying that any alternative is preferred to any other. This insight points to the possibility of treating cycles of votes as uninformative ties and using the non-conflicting remainder to determine an aggregate preference ranking.

We represent a set of votes that represent decision makers' pairwise preferences over alternatives by a {\em vote graph}. The nodes of the vote graph are the alternatives.  For each pair of alternatives $i$ and $j$, we let $q_{ij}$ denote the number of voters who prefer $i$ to $j$.  An arc $(i, j)$ is in the vote graph provided that $q_{ij}>0$. The collection of votes is said to be \emph{conflicting} if there is a directed cycle in the vote graph. A set of votes is \emph{non-conflicting} if there is no directed cycle in the vote graph. For non-conflicting sets of votes, there is an induced partial order in which alternative $i$ is preferred to alternative $j$ if and only if there is a directed path from $i$ to $j$ in the vote graph.  This partial order is the smallest one that is consistent with the votes.

For each pair $i, j$ of alternatives, we refer to the number of votes removed in which $i$ is preferred to $j$ as the flow in $(i, j)$. Our method removes votes that correspond to a flow circulation in the vote graph. That is, for each alternative $i$, the flow out of $i$ is equal to the flow into $i$. It is well known that any circulation can be decomposed into flows around cycles (see, for example, Chapter 3 of \citet{ahuja_network_1993}). 

A necessary condition for our preference aggregation approach to induce a consensus partial order is that, after the removal of a circulation, the remaining votes are non-conflicting; equivalently, the remaining vote graph is acyclic.  We refer to such circulations as \emph{maximal}.  The arc $(i, j)$ is in the remaining vote graph whenever the maximal circulation sends strictly less than $q_{ij}$ units of flow in $(i, j)$. In this case, we say that there is an \emph{aggregate preference} for $i$ over $j$ following the removal of the circulation.

There are many possible ways of removing maximal circulations from a vote graph, and thus there are many possible acyclic vote graphs that remain after their removal. When considering approaches for removing cycles, we emphasize two important properties. First, the procedure has to result in a unique partial order. Second, the procedure has to be computationally efficient. These properties are especially important in settings where equitable treatment of alternatives is an important consideration. Suppose that there was an optimal partial order in which alternative $i$ was preferred to alternatives $j$ and $k$ and another optimal partial order in which $j$ was preferred to $i$ and $k$. Then it would not be clear whether the method should report an aggregate preference for $i$ over $j$ or $j$ over $i$. And, further, in settings with many alternatives or voters, it could be computationally prohibitive to determine whether there is also a third optimal partial order for which $k$ was preferred to $i$ and $j$.

We first consider removing a \textit{maximum circulation}, which is a circulation containing a maximum total number of votes. We show that this procedure permits the efficient computation of a consensus partial order, but does not result in a unique partial order. We subsequently introduce the \emph{strong maximum circulation} as a maximum circulation whose removal returns as many aggregate preferences as possible. In other words, the remaining vote graph has as many distinct arcs as possible.  We refer to the partial order induced by this acyclic graph as a \emph{strong partial order}.  The strong partial order satisfies the desired properties of uniqueness and efficiency.  In particular, we establish the following:

\begin{enumerate}
    \item There is a unique strong partial order.
    \item The strong partial order is the union of the partial orders induced by removing any maximum circulation.
    \item One can obtain a strong maximum circulation by solving a single minimum cost flow problem.
    \item  There are optimality conditions for the strong maximum circulation problem based on a relaxation of strict complementary slackness conditions for the maximum circulation problem.
\end{enumerate}

A natural alternative approach would be to remove as few votes on cycles as possible to produce an acyclic remainder. We refer to this type of maximal circulation as a \textit{minimum maximal circulation}. Despite its intuitive appeal, this approach does not satisfy the properties that we were seeking. First, we show that it is NP-hard to find a minimum maximal circulation. Second, we show that it does not lead to a unique solution.  Worse yet, there may be two minimum maximal circulations that lead to conflicting partial orders.  That is, after removing the first minimum maximal circulation from the vote graph, alternative $i$ is preferred to alternative $j$. After removing the second one, $j$ is preferred to $i$.

We also relate our cycle removal-based approach to preference aggregation to existing optimization-based approaches, focusing on Kemeny's method \citep{kemeny_mathematics_1959}. Kemeny's method is to remove the fewest number of votes so that the remaining set of votes is non-conflicting. Equivalently, Kemeny's method is to find a partial order consistent with the largest weight-directed acyclic graph among a set of votes \citep[see, e.g.,][]{jagabathula_personalized_2022}. Although Kemeny's method also has some desirable theoretical properties, e.g., it can be justified in terms of minimizing the swap distance from a set of individual rankings \citep{fischer2016weighted}, it also suffers similar drawbacks to the problem of removing a minimum maximal circulation from the vote graph in terms of computational complexity \citep{garyjohnson,bartholdi_voting_1989} and the multiplicity of (conflicting) optimal solutions \citep{muravyov_dealing_2014, yoo_new_2021}.

The objective in Kemeny's model appears to be diametrically opposed to the objective of the strong maximum circulation model.  Kemeny's method removes as few votes as possible.  A strong maximum circulation removes as many votes as possible. However, there appears to be an interesting type of duality between the two approaches as revealed by the  Hochbaum and Levin (HL) separation-deviation model \citep[see][]{hochbaum_methodologies_2006}. Specifically, Kemeny's model can be transformed into a special case of the HL model with a non-convex ``0-1'' loss function. If this objective is approximated by a convex hinge loss function, the resulting optimization problem is the dual of the maximum circulation problem.  As such, the maximum circulation model can be viewed as the dual of the ``convexification'' of Kemeny's model.

The paper proceeds as follows. In the next section, we restate the problem of preference aggregation in terms of the isolation of an acyclic subgraph in a capacitated directed network, or vote graph. In Section \ref{sec:strong-max-circ}, we introduce the concept of a strong maximum circulation as a method for eliminating cycles. In Section \ref{sec:certif} we provide an efficient algorithm for finding a strong maximum circulation. In Section \ref{sec:HL}, we describe our method in terms of a convex relaxation of Kemeny's model and demonstrate its relationship to \cite{hochbaum_methodologies_2006}. Finally, we consider the removal of minimum maximal circulations in Section \ref{sec:minmax}.  We show that an aggregation approach based on eliminating minimum maximal circulations faces many of the same computational and normative problems as Kemeny's method.

\section{Background and Preliminaries}\label{sec:prelim}

\noindent \textbf{Vote graph and circulations.} 
A vote graph is a directed graph $G=(V,A)$ with the set of alternatives $V$ represented by the nodes (or vertices) of the graph. Let $n = |V|$ and $m=|A|$.  Each arc $(i,j)\in A$ represents the (strict) preference of $i$ over $j$ by at least one voter.  If there are $q_{ij}\geq 1$ voters that express the same preference $(i,j)$, we represent it by associating a capacity $q_{ij}$ with arc $(i,j)$.  Alternatively, one could represent the vote graph as a {\em multi-graph} with $q_{ij}$ arcs going from $i$ to $j$.   An example of a vote graph represented as multi-graph is given in Figure \ref{fig:multi-graph}, and the same vote graph represented as a weighted, or capacitated, simple graph is given in Figure \ref{fig:capacitated-graph}. Note that it is possible for both $q_{ij}$ and $q_{ji}$ to be positive, meaning that there are voters that prefer $i$ to $j$ and voters that prefer $j$ to $i$.  In that, case we have $\min \{q_{ij},q_{ji}\}$ 2-cycles of votes that contradict each other. A directed graph $G=(V,A)$ is called \textit{acyclic} if there are no directed cycles in $G$.

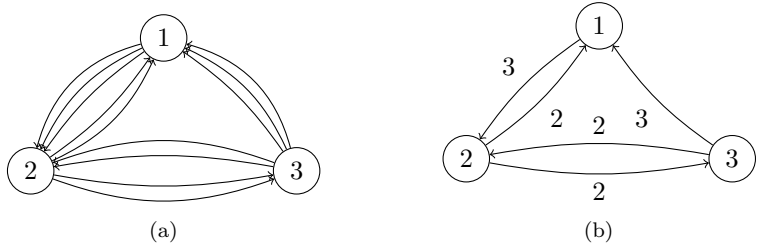
\begin{figure}[htb]
\[\hbox{
\subfloat[]{\label{fig:multi-graph}
\begin{tikzpicture}[main/.style = {draw, circle}, node distance=2.5cm] 
\node[main] (1) {$1$}; 
\node[main] (2) [below left of=1] {$2$}; 
\node[main] (3) [below right of=1] {$3$}; 
\path[bend right=10] (1) edge [->] (2);
\path[bend right=20] (1) edge [->]  (2);
\path[bend right=30] (1) edge [->]  (2);
\path[bend right=10] (2) edge [->] (1);
\path[bend right=20] (2) edge [->] (1);
\path[bend right=10] (3) edge [->]  (1);
\path[bend right=20] (3) edge [->]  (1);
\path[bend right=30] (3) edge [->]  (1);
\path[bend right=10] (2) edge [->] (3);
\path[bend right=20] (2) edge [->] (3);
\path[bend right=10] (3) edge [->] (2);
\path[bend right=20] (3) edge [->] (2);
\end{tikzpicture}
}
}
\\
\hspace{0.5in}
\hbox{
\subfloat[]{\label{fig:capacitated-graph}
\begin{tikzpicture}[main/.style = {draw, circle}, node distance=2.5cm] 
\node[main] (1) {$1$}; 
\node[main] (2) [below left of=1] {$2$}; 
\node[main] (3) [below right of=1] {$3$}; 
\path[bend right=10] (1) edge [->] node [above left] {3} (2);
\path[bend right=10] (2) edge [->] node [below right] {2} (1);
\path[bend left=10] (3) edge [->] node [below left] {3} (1);
\path[bend right=10] (2) edge [->] node [below] {2} (3);
\path[bend right=10] (3) edge [->] node [above] {2} (2);
\end{tikzpicture}
}
}
\\
\]
    \caption{The representation of the vote graph as a: (a) multi-graph,  (b) capacitated graph.}
    \label{fig:capacitated-multi}
\end{figure}

We denote vectors by boldface characters. Let the vector of capacities of the arcs of $A$ in the vote graph be $\q$, where $\q \in \mathbb{R}_+^m$. In our discussion, we typically assume that $q_{ij}$ represents the count of voters who prefer $i$ to $j$ and is therefore integral. However, as we show in the next section, it may be desirable to remove circulations of votes that leave fractional remainders from which the consensus partial order is then induced. As such, we permit capacities $\q$ to be non-integral. Note further that the all results in the paper straightforwardly generalize to settings when $\q$ is initially fractional, as discussed in Subsection \ref{sec:strong-max-circulations} and Subsection \ref{sec:more-efficient}.


An alternative, compact representation of the vote graph is denoted by $G[V,\q ]$ with the interpretation that this graph is $G=(V,A)$ where $A= A(\q)$ and $A(\q)=\{ (i,j): q_{ij}>0 \}$. 

A  (flow) vector $\x \in \mathbb{R}_+^m$ is said to be a {\em circulation} in $G[V,\q ]$ if it satisfies  that $0 \le x_{ij} \le q_{ij}$ for all $(i,j)\in A(\q)$, and $\sum_{j: (i,j) \in A} x_{ij}-\sum_{j:  (j,i) \in A} x_{ji}=0$ for all $i\in V$. That is, the incoming flow into each node is equal to the outgoing flow from that node. 

A circulation $\x'$ is  {\em maximal} if there is no circulation $\x'' \ne \x'$ such that $\x'' \ge  \x'$.  
\begin{observation}
\label{obs:max-circ-acyclic}
For  a maximal circulation $\x'$ in $G[V,\q ]$, the graph $G[V,\q -\x' ]$ is acyclic.  Equivalently, $G' = (V,A(\q -\x'))$ is acyclic.
\end{observation}
This follows since if $G'$ contains a cycle, then $\x'$ can be appended by that cycle, contradicting its being maximal.

We next define three specific types of maximal circulations: (1) minimum maximal, (2) maximum, and (3) strong maximum.
\begin{enumerate}
\item A {\em minimum maximal circulation} is a maximal circulation that contains the smallest number of arcs counted with their multiplicity. That is, it is a circulation $\x'$ that minimizes $\sum _{(i,j)\in A} x_{ij}$ among all maximal circulations.
\item A {\em maximum circulation} is a circulation that contains the largest number of arcs counted with their multiplicity. That is, it is a circulation $\x'$ that solves the following:
\begin{align}
    \label{eqn:max-circ}
      \max & \sum_{(i,j)\in A}  x_{ij}\\
      \text{(Maximum Circulation) \  \ \ subject to } & \sum_{j: (i,j) \in A} x_{ij}-\sum_{j:  (j,i) \in A} x_{ji}=0  \ \ \ \forall i\in V \nonumber\\
     & 0 \le x_{ij} \le q_{ij} \ \ \  \forall (i,j) \in A. \nonumber
\end{align}
By construction, all maximum circulations are also maximal.
\item A {\em strong maximum circulation}, denoted $\x^*$, is a maximum circulation for which the minimum number of arcs in $G$ are at capacity. That is, for any other maximum circulation $\x'$,  $|A(\q - \x^*)| \ge  |A(\q - \x')|$.\\
\end{enumerate}

\noindent\textbf{Acyclic graphs, transitive closures, and partial orders}.
For an acyclic graph $G = (V, A(\q))$, the {\em transitive closure} of $G$ contains all pairs $(i,j)$ such that there is a directed path in $G$ from $i$ to $j$.  We denote the arc set of the transitive closure as $A^{TC}(\q)$ and denote the transitive closure of $G$ as $G^{TC} = (V, A^{TC}(\q))$. When $\q$ is clear from context, we refer to the transitive closure compactly as $A^{TC}$.

A {\em partial order} is a transitive and antisymmetric\footnote{A binary relation $R$ on a set $X$ is antisymmetric if for all $i,j$ in $X$, $i R j$ and $j R i$ implies $i=j$.} relation among the elements of a set.  We denote the relation using the symbol ``$\succ$.''  Associated with each acyclic graph $G = (V, A)$ is a {\em partial order}, obtained as follows: for elements $i, j \in V$, $i \succ j$ if and only if $(i,j) \in A^{TC}$.

If $G$ is acyclic, we sometimes use $A^{TC}$ to represent both the transitive closure and the corresponding partial order. Accordingly, we use $A^{TC}(\q - \x')$ to represent the transitive closure of the arcs that remain following the elimination of maximal circulation $\x'$ and the partial order induced by the removal of $\x'$. 

Partial orders $P$ and $Q$ defined on the same set of elements are said to \emph{conflict} if there are elements $i$ and $j$ such that $i \succ_P j$ and $j \succ_Q i$.  Two partial orders defined on the same set of elements are said to be \emph{non-conflicting} if they do not conflict.\\


\noindent \textbf{Strong maximum circulations, strong partial order, strong arcs, and strong subgraph.}
We will show that for any two strong maximum circulations $\x'$ and $\x''$, $A(\q - \x' ) = A(\q - \x'')$, and thus   $A^{TC}(\q - \x')$  $= A^{TC}(\q - \x'')$.   That is, the set of arcs that remains following the elimination of a strong maximum circulation $\x^*$ is independent of the strong maximum circulation. 

We refer to $A^{TC}(\q-\x^*)$ as the \textit{strong partial order}, and denote it more simply as $A^{SP}$. We say that an arc $(i, j)$ is \emph{strong} if there is a maximum circulation $\x'$ such that $x'_{ij} < q_{ij}$, and thus $(i, j) \in A(\q - \x')$. The \emph{strong subgraph} of $A$ is the set $A^* = \{(i, j) : (i, j) \text{ is strong} \}$.  We will show that the transitive closure of $A^*$ is equal to $A^{SP}$.\\

\noindent \textbf{Network Flow.} We next introduce some well-known results from the network flow literature that will be used in the computational analysis of our algorithm. A more detailed treatment is available in \citet{ahuja_network_1993}.\\

\noindent{\bf The minimum cost network flow problem (MCNF).} The minimum cost network flow problem is defined on a directed graph $G=(V,A)$ where every arc $(i,j)\in A$ has a nonnegative arc capacity $q_{ij}$ and cost $c_{ij}$ per unit flow. 
Every node $i\in V$ has a {\em supply} value $b_i $ associated with it that could be positive, negative or zero, and such that $\sum _{i\in V} b_i =0$.  A vector $\x \in \mathbb{R}_+^m$ is said to be a {\em feasible flow} if $0 \le x_{ij} \le q_{ij}$ for all $(i,j)\in A(\q)$, and for all $i\in V$, $ \sum_{j: (i,j) \in A} x_{ij}-\sum_{j:  (j,i) \in A} x_{ji}=b_i$.    The formulation of MCNF is given in (\ref{eqn:MCNF}):

\begin{align}
    \label{eqn:MCNF}
      \min & \sum_{(i,j)\in A} c_{ij} x_{ij}\\
      \text{(MCNF) \  \ \ subject to } & \sum_{j: (i,j) \in A} x_{ij}-\sum_{j:  (j,i) \in A} x_{ji}=b_i  \ \ \ \forall i\in V \nonumber\\
     & 0 \le x_{ij} \le q_{ij} \ \ \  \forall {(i,j) \in A.} \nonumber
\end{align}

The maximum circulation problem is a special case of MCNF in which $b_i=0$ for all $i\in V$, the minimization is replaced by maximization, and each cost is equal to $1$.  (This is equivalent to a minimization problem with costs of $-1$.)

\medskip

{\bf Complexity of solving MCNF and maximum circulation}.  Using the successive approximation algorithm of \cite{goldberg_solving_1987}, the MCNF problem is solved in running time $O(nm \log n \log(nC))$, where $n=|V|$ $m=|A|$ and $C=\max _{(i,j)\in A} |c_{ij}|$.   Goldberg and Tarjan also provide the running time bound of $O(n^{5/3}m^{2/3} \log nC)$, which improves on $O(nm \log n \log(nC))$ when $ \frac{n^2}{m} < \log^3 n$.

For the maximum circulation problem, $C=1$.  
Therefore, the running time to solve the maximum circulation is $O(nm \log^2 n )$.

\section{Maximum Circulations and Strong Maximum Circulations} \label{sec:strong-max-circ}

In this section, we analyze preference aggregation methods that are based on eliminating maximum circulations.  

Specifically, the section includes the following results: 
\begin{enumerate}

\item Let $A^*$ denote the strong subgraph. If $\x^*$ is any strong maximum circulation, then $A(\q - \x^*) = A^*$, (Subsection \ref{sec:max-circ-non-conflict}).  This implies that $A^*$ is acyclic. 

\item An algorithm that computes a strong maximum circulation by solving $m+1$ maximum circulation problems.  (Subsection \ref{sec:strong-max-circulations}).
\end{enumerate}

\subsection{ Generating a strong maximum circulation based on the acyclic set $A^*$}\label{sec:max-circ-non-conflict}

Recall that the strong subgraph $A^*$ includes each arc that is not at capacity in some maximum circulation.  Here, we show how to obtain a strong maximum circulation based on $A^*$.  We also establish that $A^*$ is acyclic.

We first start with the special case in which the vote graph is Eulerian. The vote graph is \emph{Eulerian} if $ \sum_{j: (i,j) \in A} q_{ij}=\sum_{j:  (j,i) \in A} q_{ji}$. The following lemma is a well known result about Eulerian networks.

\begin{lemma}
    If the vote graph $G = (V, A(\q))$ is Eulerian, then the flow $\x^* = \q$ is the unique maximum circulation, and $A^* = \emptyset$. If the vote graph is not Eulerian, then $\x^* = \q$ is not a circulation, and $A^* \neq \emptyset$.
\end{lemma}

\begin{corollary}
    Let $A^*$ denote the strong subgraph of the vote graph $G = (V, A(\q))$.  Then $A^* = \emptyset$ if and only if the vote graph is Eulerian.  
\end{corollary}

\begin{lemma}
\label{lem:strong-partial}

Let $G = (V, A(\q))$ be a non-Eulerian vote graph. Let $A^*$ denote the strong subgraph of $G$.   
For each arc $a \in A^*$, let $\x^a$ be a maximum circulation for which the flow in arc $a$ is less than $q_a$.  Let $\x ^* =\frac {\sum_{a \in A^*}\x ^a}{|A^*|}$.  Then 
\begin{enumerate}
    \item $A^* = A(q - \x^*)$;
    \item $\x^*$ is a strong maximum circulation, and;
    \item $A^*$ is acyclic.
\end{enumerate}
\end{lemma}

\begin{proof}
$\x^*$ is the convex combination of maximum circulations, and is thus a maximum circulation.  No arc of $A^*$ is at its upper bound in $\x^*$.  Therefore, $A^* \subseteq A(q - \x^*)$.  Every arc in $A\setminus A^*$ is at its upper bound in every maximum circulation.  Therefore, $(A\setminus A^*) \cap A(q - \x^*) = \emptyset$.  Hence,
$A(q - \x^*) \subseteq A^*$.  It follows that $A^* = A(q - \x^*)$.  Since  $\x^*$ is a maximum circulation, $ A(q - \x^*)$ must be acyclic. And since $A^* = A(q - \x^*)$ it follows that $\x^*$ is a strong maximum circulation.  
 \end{proof}

We now show that the arc set $A^*$ is equal to the set of arcs remaining after removing any strong maximum circulation from the graph. It is therefore unique.

\begin{corollary}\label{lem:strong-max}
      
    A maximum circulation $\x'$ in the vote graph $G = (V, A(\q))$ is strong if and only if $A(\q-\x') = A^*$. 
\end{corollary}

    \begin{proof}

Let $x^*$ be defined as in Lemma \ref{lem:strong-max}.  Suppose that $\x'$ is a strong maximum circulation.  It follows that
    $|A(\q-\x')| = |A^*| = |A(\q-\x^*)|$.   Moreover, $A(\q-\x') \subseteq A^*$.  Therefore, $A(\q-\x') = A^*$.  
\end{proof}

\begin{corollary}\label{lem:non-conflictx}
For any two maximum circulations $\x'$ and $\x''$, the partial orders induced by $A(\q - \x')$ and $A(\q - \x'')$ are non-conflicting.
\end{corollary}
\begin{proof}
Suppose that for two maximum circulations $\x'$ and $\x''$ there is a conflict. Therefore there is a pair $i$ and $j$ such that $i \succ_{\x'} j$ and $j \succ_{\x''} i$. 
By definition, $ A^*$ contains both arcs $(i,j)$ and $(j,i)$ which contradicts that $ A^*$ is acyclic.
\end{proof}

Recall that $A^{SP}$ is the strong partial order, which is the partial order induced by $A^*$, or equivalently by $A(\q-\x')$, where $\x'$ is any strong maximum circulation. According to Corollary \ref{lem:strong-max}, if $\x'$ is a maximum circulation and $A(\q - \x') = A^*$, then $\x'$ is a strong maximum circulation. Corollary \ref{lem:strong-max} would not be true if one relaxed the initial condition and only required that $\x'$ is a feasible circulation.  In the following example, $\x'$ is a feasible circulation such that $A(\q - \x') = A^*$; but $\x'$ is not a maximum circulation. The example consists of a graph with node set $\{1, 2, 3\}$.  There are arcs $(1, 2)$, $(2, 3)$, and $(1, 3)$, each with capacity 2.  There is also an arc $(3, 1)$ with capacity 1.  The maximum circulation $\x^*$ has  $x^*_{1,2} = x^*_{2,3} = x^*_{3,1} = 1$ and $x^*_{1, 3} = 0$. It is a strong maximum circulation, and $A^* = A(q - \x^*) = \{(1, 2), (1, 3), (2, 3)\}$.  There is a also a feasible non-maximum circulation $\x'$ with $x'_{1,3} =  x'_{3,1} = 1$ and $x'_{1, 2} = x'_{2,3} = 0$. 
 Note that $A^* = A(q - \x') = \{(1, 2), (1, 3), (2, 3)\}$.    

\subsection{Constructing a Strong Maximum Circulation} 
\label{sec:strong-max-circulations}
 
The development of the strong maximum circulation $\x^*$ in Lemma \ref{lem:strong-max} is not a constructive algorithm because it does not construct the flows $x^a$ for $a \in A^*$.  The next algorithm efficiently constructs these flows one arc $a$ at a time checking whether a maximum circulation flow has the flow on the arc strictly less than $q_a$.  For simplicity, it is initially assumed that all data is integral.\\

{\noindent \underline{\sf Algorithm 1}}:  \\
01.  \hspace{2 em}  $z^* :=$  optimum objective value for the maximum circulation problem (\ref{eqn:max-circ}). \\ 
02.  \hspace{2 em}    $A^*:=\emptyset$
\\
03.  \hspace{2 em}  {\bf for each} arc $(i, j)\in A$ {\bf do};\\ 
04.  \hspace{4 em}  let $z_{ij}$ be the optimal objective value for the maximum circulation problem\\ 
. \hspace{7 em} with the added constraint that $x_{ij} \le q_{ij} - 1$.\footnote{Subtracting 1 works because the vector $\q$ is integral, and there is always an integer optimum flow if all data is integral.}\\
05.  \hspace{4 em} {\bf if} $z_{ij} = z^*$, {\bf then} $A^*:=A^* \cup \{(i,j)\}$ \\
06.   \hspace{2 em}   {\bf endfor}\\
07. \hspace{2 em} $A^{SP}$ is the transitive closure of $A^*$\\

The algorithm can be easily modified to work for data that is not integral. For non-integral $\q$,  the constraint in line 04 of Algorithm 1 would be modified to ``$x_{ij} \le q_{ij} - \epsilon$'', where $\epsilon$ is chosen to be sufficiently close to 0. Additionally, note that the algorithm may deliver $A^*$ empty if the vote graph is Eulerian.

Algorithm 1 determines $A^*$ by solving $m+1$ maximum circulation problems, including one for determining the optimal objective value. Each maximum circulation problem is solvable in $O(nm \log^2 n)$ time. In Section \ref{sec:certif}, we show how to determine $A^*$ by solving just one minimum cost flow problem. 

Using Lemma \ref{lem:strong-partial}, the output of Algorithm 1 can deliver a strong maximum circulation. For each $a \in A^*$, let $\x^a$ be the maximum circulation found by the algorithm in which the flow in $a$ is less than its capacity. 
Assuming that the vote graph is not Eulerian, and thus $A^* \neq \emptyset$, the average of these $|A^*|$ maximum circulations, $\x ^* =\frac {\sum_{a \in A^*}\x ^a}{|A^*|}$, is a strong maximum circulation.

Even though $A^*$ and  $A^{SP}$ are uniquely determined, 
there may be an infinite number of different strong maximum circulations. For example, consider any flow $\x'$ obtained as follows:

\begin{align}
    \label{eqn:strong-circ2}
    \x'= \sum_{a \in A^*} \lambda_a \x ^a.
\end{align}

\noindent where $\{ \lambda _a \}_{a \in A^*}$ 
is any strictly positive vector whose components sum to 1. Then $\x'$ is a strong maximum circulation.

\subsubsection{An Illustrative Example}\label{sec:example}
We next illustrate the construction of the strong partial order and the strong maximum circulation for a graph $G[V,\q ]$.  Recall that the arcs in the graph $G$ are $A=A(\q) = \{ (i,j): i,j\in V {\rm and\ } q_{ij} >0\}$.  We represent the capacities vector $\q$ as an $n \times n$ matrix where only the positive entries correspond to the arcs of $G=[V,\q ]$.

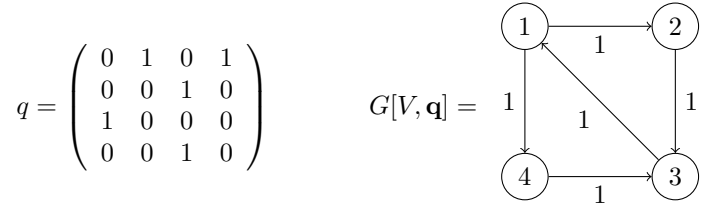
\[q = 
 \left( \begin{array}{cccc}
0 & 1 & 0 & 1 \\
0 & 0 & 1 & 0 \\
1 & 0 & 0 & 0\\
0 & 0 & 1 & 0
\end{array} \right)
\hspace{0.5in}
G[V,\q]=
\vcenter{\hbox{
\begin{tikzpicture}[main/.style = {draw, circle}, node distance=2cm] 
\node[main] (1) {$1$}; 
\node[main] (2) [right of=1] {$2$}; 
\node[main] (3) [below of =2] {$3$}; 
\node[main] (4) [below of =1] {$4$}; 
\path (1) edge [->] node [below] {1} (2);
\path (2) edge [->] node [right] {1} (3);
\path (4) edge [->] node [below] {1} (3);
\path (1) edge [->] node [left] {1} (4);
\path (3) edge [->] node [below left] {1} (1);
\end{tikzpicture}
}}\]

In this graph there are only $5$ arcs and two integer-valued maximum circulations $\x'$ and $\x''$, each of value $3$. The circulation $\x'$ is obtained by sending one unit of flow on the cycle 1-2-3-1; the circulation $\x''$ is obtained by sending one unit of flow on the cycle 1-4-3-1.  These circulations and the remaining (residual flows) graphs $G[V,\q-\x']$ and $G[V,\q-\x'']$ are:

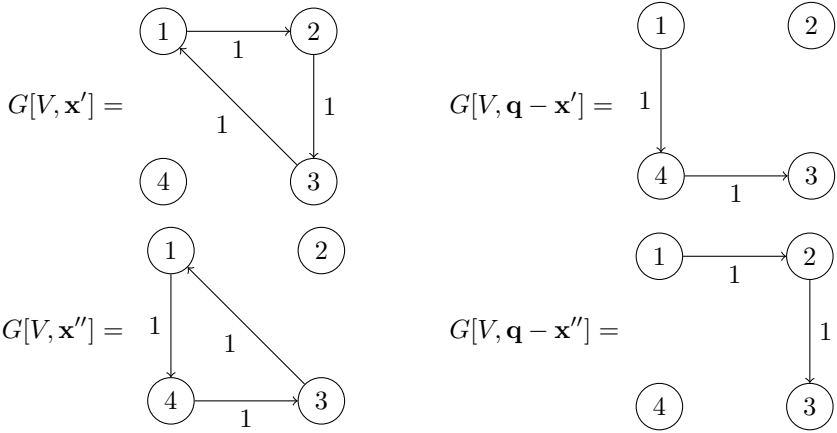
\[
G[V,\x']=
\vcenter{\hbox{
\begin{tikzpicture}[main/.style = {draw, circle}, node distance=2cm] 
\node[main] (1) {$1$}; 
\node[main] (2) [right of=1] {$2$}; 
\node[main] (3) [below of =2] {$3$}; 
\node[main] (4) [below of =1] {$4$}; 
\path (1) edge [->] node [below] {1} (2);
\path (2) edge [->] node [right] {1} (3);
\path (3) edge [->] node [below left] {1} (1);
\end{tikzpicture}
}}
\hspace{0.5in}
G[V,\q-\x']=
\vcenter{\hbox{
\begin{tikzpicture}[main/.style = {draw, circle}, node distance=2cm] 
\node[main] (1) {$1$}; 
\node[main] (2) [right of=1] {$2$}; 
\node[main] (3) [below of =2] {$3$}; 
\node[main] (4) [below of =1] {$4$}; 
\path (4) edge [->] node [below] {1} (3);
\path (1) edge [->] node [left] {1} (4);
\end{tikzpicture}
}}
\]

\[
G[V,\x'']=
\vcenter{\hbox{
\begin{tikzpicture}[main/.style = {draw, circle}, node distance=2cm] 
\node[main] (1) {$1$}; 
\node[main] (2) [right of=1] {$2$}; 
\node[main] (3) [below of =2] {$3$}; 
\node[main] (4) [below of =1] {$4$}; 
\path (4) edge [->] node [below] {1} (3);
\path (1) edge [->] node [left] {1} (4);
\path (3) edge [->] node [below left] {1} (1);
\end{tikzpicture}
}}
\hspace{0.5in}
G[V,\q-\x'']=
\vcenter{\hbox{
\begin{tikzpicture}[main/.style = {draw, circle}, node distance=2cm] 
\node[main] (1) {$1$}; 
\node[main] (2) [right of=1] {$2$}; 
\node[main] (3) [below of =2] {$3$}; 
\node[main] (4) [below of =1] {$4$}; 
\path (1) edge [->] node [below] {1} (2);
\path (2) edge [->] node [right] {1} (3);
\end{tikzpicture}
}}
\]

Consider next the circulation $\x^* = (\x' + \x'')/2$ and the respective remaining graph after removing $\x^*$, $G[V,\q-\x ^*]$:
 
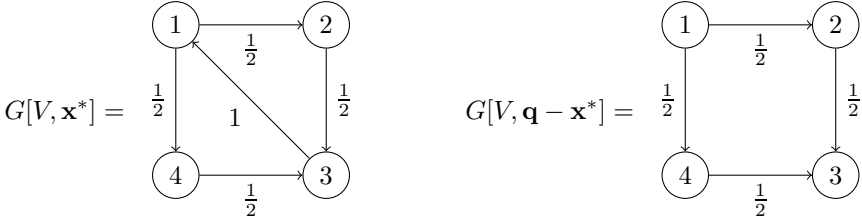
\[
G[V,\x^*]=
\vcenter{\hbox{
\begin{tikzpicture}[main/.style = {draw, circle}, node distance=2cm] 
\node[main] (1) {$1$}; 
\node[main] (2) [right of=1] {$2$}; 
\node[main] (3) [below of =2] {$3$}; 
\node[main] (4) [below of =1] {$4$}; 
\path (1) edge [->] node [below] {$\frac{1}{2}$} (2);
\path (2) edge [->] node [right]  {$\frac{1}{2}$} (3);
\path (4) edge [->] node [below] {$\frac{1}{2}$} (3);
\path (1) edge [->] node [left] {$\frac{1}{2}$} (4);
\path (3) edge [->] node [below left] {1} (1);
\end{tikzpicture}
}}
\hspace{0.5in}
G[V,\q- \x ^*]=
\vcenter{\hbox{
\begin{tikzpicture}[main/.style = {draw, circle}, node distance=2cm] 
\node[main] (1) {$1$}; 
\node[main] (2) [right of=1] {$2$}; 
\node[main] (3) [below of =2] {$3$}; 
\node[main] (4) [below of =1] {$4$}; 
\path (1) edge [->] node [below] {$\frac{1}{2}$} (2);
\path (2) edge [->] node [right]  {$\frac{1}{2}$} (3);
\path (4) edge [->] node [below] {$\frac{1}{2}$} (3);
\path (1) edge [->] node [left] {$\frac{1}{2}$} (4);
\end{tikzpicture}
}}
\]

Note that the graph $G[V,\q- \x ^*]$ is acyclic, and therefore the set of arcs $A(\q- \x ^*)$ forms an acylic graph. That follows since $\x ^*$ is a maximal circulation.  However, note that while the number of arcs in $A(\q-\x')$ is $2$ and the number of arcs in $A(\q-\x'')$ is $2$, the number of arcs in $A(\q-\x ^*)$ is $4$ and hence strictly greater than the number of arcs in the acyclic graph induced by the integer-valued circulations. Following Lemma \ref{lem:strong-max}, $A(\q-\x^*)=A(\q-\x') \cup A(\q-\x'')$. Finally, by adding the transitive closure arcs to $A(\q- \x ^*)$, we get the strong partial order $A^{SP}$.  

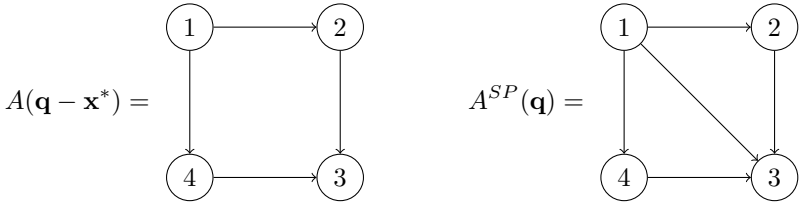
\[
A(\q-\x ^*)=
\vcenter{\hbox{
\begin{tikzpicture}[main/.style = {draw, circle}, node distance=2cm] 
\node[main] (1) {$1$}; 
\node[main] (2) [right of=1] {$2$}; 
\node[main] (3) [below of =2] {$3$}; 
\node[main] (4) [below of =1] {$4$}; 
\path (1) edge [->] node [below] {} (2);
\path (2) edge [->] node [right]  {} (3);
\path (4) edge [->] node [below] {} (3);
\path (1) edge [->] node [left] {} (4);
\end{tikzpicture}
}}
\hspace{0.5in}
A^{SP}(\q)=
\vcenter{\hbox{
\begin{tikzpicture}[main/.style = {draw, circle}, node distance=2cm] 
\node[main] (1) {$1$}; 
\node[main] (2) [right of=1] {$2$}; 
\node[main] (3) [below of =2] {$3$}; 
\node[main] (4) [below of =1] {$4$}; 
\path (1) edge [->] node [below] {} (2);
\path (2) edge [->] node [right]  {} (3);
\path (4) edge [->] node [below] {} (3);
\path (1) edge [->] node [left] {} (4);
\path (1) edge [->] node [below left] {} (3);
\end{tikzpicture}
}}
\]

\section{Computing a Strong Circulation with One Maximum Circulation Flow }
\label{sec:certif}

Algorithm 1, described
in Section \ref{sec:strong-max-circ}, computes the acyclic graph $A^*$, which in turn can be used to construct a strong maximum circulation and a strong partial order.
Algorithm 1 requires the solution of at most $m+1$ maximum circulation problems.  In this section, we introduce the \emph{perturbation algorithm}, which determines $A^*$ by solving a single minimum cost network flow (MCNF) problem.  Specifically, the perturbation algorithm returns a strong maximum circulation, which is then used to find $A^*$ by removing the arcs at capacity.

The perturbation algorithm relies on the optimality conditions that are shown to characterize strong maximum circulations---the {\emph{strong complementary slackness conditions}}, as defined in Subsection  \ref{sec:strong}. The formulation of the associated MCNF problem and the perturbation algorithm are subsequently given in Subsection \ref{sec:more-efficient}.

\subsection{Strong (Rather than Strict) Complementary Slackness}
\label{sec:strong}
We define here the ``strong complementary slackness" conditions that are proved to characterize strong maximum circulations. These conditions are more restrictive than the usual complementary slackness conditions, but less restrictive than the conditions that are known as strict complementary slackness.  We will subsequently show that the strong maximum circulation derived from the MCNF solution satisfies strong complementary slackness.  It does not necessarily satisfy strict complementary slackness.

Among all optimal solutions for the maximum circulation problem (\ref{eqn:max-circ}), we want a strong maximum circulation; that is, we want a maximum circulation with as many arcs as possible with flow less than than their capacity.  Suppose that we are given a circulation $\x'$.  One can identify that $\x'$ is maximum using a form of LP duality such as the complementary slackness conditions. It is shown next that one can conclude that $\x'$ is strong if it satisfies the ``strong complementary slackness conditions.''

The dual of the maximum circulation problem (\ref{eqn:max-circ}) is:

\begin{align}
    \label{eq:dual-program}
      \min & \sum_{(i,j)\in A} q_{ij}z_{ij} \\
     \text{subject to } & y_i -y_j +z_{ij} \ge 1 \,\, \forall (i,j)\in A \nonumber\\
     &z_{ij} \ge 0 \,\, \forall (i,j)\in A. \nonumber
\end{align}

Suppose that $\x'$ is a feasible circulation for (1) and that $(\y', \z')$ is feasible for (\ref{eq:dual-program}).  We will express the strong optimality conditions based entirely on the vectors $\x'$ and $\y'$.  We may do so without loss of generality because we assume that for all $(i,j) \in A$, $z'_{ij}=\max \{ y'_j-y'_i +1, 0\}$.  These conditions are satisfied by any optimal solution to (\ref{eq:dual-program}).

Recall that we use the notation $c_{ij}$ to represent the cost per unit flow on arc $(i,j)$ in a MCNF problem (\ref{eqn:MCNF}) and that $c_{ij} = 1$ for all $(i, j) \in A$ for the primal maximum circulation problem (\ref{eqn:max-circ}). We say that $\x'$ and $\y'$ satisfy \emph{strong complementary slackness} for the MCNF problem if the following is true:

\begin{enumerate}
\item  For each $(i,j) \in A$, if $x'_{ij} = 0$, then $c_{ij} - y'_i + y'_j  \le  0$.
\item  For each $(i,j) \in A$, if $0 < x'_{ij} < q_{ij}$, then $c_{ij} - y'_i + y'_j  =  0$.
 \item  For each $(i,j) \in A$, if $x'_{ij} = q_{ij}$, then $c_{ij} - y'_i + y'_j  >  0$.

\end{enumerate}

The above conditions are more restrictive than the standard complementary slackness conditions because of (3).    The usual complementary slackness condition is: ``if $x'_{ij} = q_{ij}$, then $c_{ij} - y'_i + y'_j  \ge  0$''.  The above conditions are also less restrictive than {\em strict} complementary slackness conditions, which require,  in addition to (2) and (3), to replace (1) by:  ``For each $(i,j) \in A$, if $x'_{ij} = 0$, then $c_{ij} - y'_i + y'_j  <  0$". 

If $\x'$ and $\y'$ satisfy strong complementary slackness, then $\y'$  ``certifies'' that $\x'$ is a strong maximum circulation, as stated in the following lemma.

\begin{theorem} \label{lem:comp-slack-strong}
  Suppose that $\x'$ is a feasible circulation for (\ref{eqn:max-circ}), and that $\y'$ is feasible for (\ref{eq:dual-program})---i.e., the dual of (\ref{eqn:max-circ}).  If $\x'$ and $\y'$ satisfy strong complementary slackness, then $\x'$ is a strong maximum circulation. 
\end{theorem}
\begin{proof}
    Assume that $\x'$ and $\y'$ satisfy strong complementary slackness. Then they also satisfy complementary slackness, and are thus optimal for the respective primal (\ref{eqn:max-circ}) and dual (\ref{eq:dual-program}) problems. Therefore, $\x'$ must be a maximum circulation and $\y'$ is a dual optimal solution.
    To prove that $\x'$ is a strong maximum circulation, 
    we show next that for any other maximum circulation  $\x''$, if $x'_{ij}=q_{ij}$ for any arc $(i,j) \in A$, then $x''_{ij}=q_{ij}$.

By the strong complementary slackness satisfied by
$\x'$ and $\y'$: if $x'_{ij}=q_{ij}$, then $c_{ij} - y'_i + y'_j  >  0$.  Since $\x''$ is a maximum circulation and $\y'$ is a dual optimal solution, it follows that $\x''$ and $\y'$ satisfy complementary slackness.   Therefore, $c_{ij} - y'_i + y'_j  >  0$, implies that $x''_{ij}=q_{ij}$.   
\end{proof}


Our focus here is on strong complementary slackness, which is a relaxation of strict complementary slackness. In fact, any linear program that has a primal and dual solution will have solutions that satisfy strict complementary slackness, as proved by \cite{goldman1956theory}. 

\begin{theorem} \label{lem:strict_comp_slackness}
  Consider the linear program $\min \{cx : Ax \le b, 0 \le x \le u\}$. If the linear program has a feasible solution and if the dual linear program also has a feasible solution, then there is are optimal solutions $\x^*$ and $\y^*$ for the primal and dual problems such that the pair $(\x^*, \y^*)$ satisfies strict complementary slackness. 
\end{theorem} 

Theorem \ref{lem:strict_comp_slackness} can be applied to prove the following.

\begin{corollary}
    \label{lem:comp-slack-converse}
  Suppose that $\x'$ is a strong maximum circulation.  Then there exists a vector $\y'$ that is feasible for (\ref{eq:dual-program}) such that the pair $(\x', \y')$ satisfies strong complementary slackness. 
\end{corollary} 

\begin{proof}
     By \cite{goldman1956theory} there must be a maximum circulation $\x^*$ and a dual feasible solution $\y^*$ such that $(\x^*, \y^*)$ satisfies strict complementary slackness.  Since $\x'$ is also maximum, $(\x', \y^*)$ satisfies complementary slackness. Because $\x'$ is a strong maximum circulation,  $A^* = A(q - \x')$.  Therefore, for all $(i, j) \in A\backslash A^*$, $x'_{ij} = q_{ij}$ and therefore $x^*_{ij} = q_{ij}$.  Let $w^*_{ij} = 1 - y^*_i + y^*_j$. By the strict complementary slackness of $(\x^*, \y^*)$, it follows that $w^*_{ij} > 0$, and thus $(\x', \y^*)$ satisfies strong complementary slackness.  
\end{proof}


Therefore, one can certify a maximum circulation $\x'$ as strong by obtaining a dual solution $\y'$ such that the pair  $(\x', \y')$ satisfies strong complementary slackness.

\subsection {The ``Perturbation" Algorithm}
\label{sec:more-efficient}
We present here the ``Perturbation" algorithm that
finds a strong maximum circulation $\x^*$ as well as the set $A^*$ and the induced strong partial order $A^{SP}$ by solving a single maximum (or minimum) cost flow problem derived from a ``perturbation" of the maximum circulation problem.   That solution also produces a dual vector $\y^*$ such that $\x^*$ and $\y^*$ satisfy strong complementary slackness.

The perturbation algorithm relies on solving the following maximum utility (negative cost) flow problem with two sets of decision variables $\w , \v \in \mathbb{R}^m$.  Recall that $m=|A|=|A(\q)|$, the number of arcs in the vote graph $G=(V,A)$.  Our network flow formulation (\ref{eqn:max_xy}) is defined on a graph $G=(V,A^{(2)})$ where $A^{(2)}$ contains 2 copies of each arc $(i,j) \in A$.  The flow in the first copy of $(i, j)$ will be denoted as $w_{ij}$ with capacity upper bound $q_{ij} - \epsilon$, and the flow in the second copy of $(i,j)$ will be denoted as $v_{ij}$ with capacity upper bound $\epsilon$.  

The perturbation algorithm can be applied for vote graphs with integral or non-integral $\q$. If $q_{ij}$ is integral for all arcs $(i,j) \in A$, then we can choose $\epsilon = 1/(m+1)$.  Otherwise, we can choose $\epsilon = 1/( (m+1)\cdot LCD)$, where LCD is the least common denominator of the values in $\q$.

We let $c_{ij} = 1$ for all arcs $(i,j) \in A$.  

 \begin{align}
    \label{eqn:max_xy}
      \max \,\,\,\,\ & f(\w, \v) =   \sum_{(i,j)\in A} c_{ij} (w_{ij}+v_{ij})
      -\frac{1}{m+1} \sum_{(i,j)\in A} v_{ij}\\
      \text{subject to } & \sum_{j: (i,j) \in A} (w_{ij}+v_{ij})-\sum_{j:  (j,i) \in A}( w_{ji}+v_{ji})=0  \,\,\ \forall i\in V, \nonumber\\     
     & 0\leq w_{ij} \leq q_{ij} - \epsilon \nonumber
     \,\,\, \text{ for }(i,j) \in A, \\
     & 0\leq v_{ij} \leq \epsilon \nonumber
     \,\,\,\, \text{ for }(i,j) \in A. 
\end{align}

\medskip

We note that model \ref{eqn:max_xy} is a maximum cost flow problem. It can be converted to a minimum cost flow problem by replacing the objective function by its negative.

The perturbation algorithm consists of first finding an optimal solution $(\w', \v')$ for problem (\ref{eqn:max_xy}), and then finding an optimal solution $(\y', \z')$ for the dual of (\ref{eqn:max_xy}).  The complexity of finding an optimal solution to MCNF, using the successive approximation algorithm of \cite{goldberg_solving_1987}, is $O(nm \log (n^2/m) \log(nC))$, where $n=|V|$ $m=|A|$ and $C=\max _{(i,j)\in A} |c_{ij}|$.  Here we scale the objective function  
by $m+1$ so all coefficients are integer, and $C=m+1$. Observing that $m=O(n^2)$ the complexity of solving (\ref{eqn:max_xy}) is $O(nm \log(n^2/m)\log n )$ -- the same complexity as that of solving one maximum circulation problem.
To derive the optimal dual solution, the standard procedure is to construct a graph where the arcs are all the residual arcs with respect to the optimal flow, with their respective costs, and to add a source node with arcs of length $0$ to all nodes in the graph.  In this constructed graph the length of the shortest paths from the source node, correspond to the optimal dual values.  Hence, the optimal dual solution $\y'$ can be obtained in $O(mn)$ time, using, e.g., the Bellman-Ford algorithm. This time to determine the optimal dual solution is dominated by the run time required to find the optimal flow.

The perturbation algorithm outputs $\x' = \w' + \v'$.  We show next that $\x'$ is a strong maximum circulation and that $(\x', \y')$ satisfies strong complementary slackness for the maximum circulation problem.

\begin{lemma}
\label{equivalent}
Let $(\w', \v')$ be an optimal basic solution for the maximum cost network flow problem (\ref{eqn:max_xy}) and $(\y',\z')$ optimal for the dual linear program.   Then $\x' = \w' + \v'$ is a strong maximum circulation and $A'= \{(i, j): v'_{ij} = 0\}$ satisfies $A' = A^*$.  
Moreover, $\x'$ and $\y'$ satisfy strong complementary slackness for the maximum circulation problem  (\ref{eqn:max-circ}). 
\end{lemma}
\begin{proof}

We first consider a solution to (\ref{eqn:max_xy}) that is derived from a strong maximum circulation.  Let $\x ^*$ be the strong maximum circulation calculated by averaging the flows from Algorithm 1.  Let
$A^* = A(q - \x^*)$ denote the arcs whose flow are not at capacity in $\x ^*$.  By Lemma \ref{lem:strong-max}, $A^*$ is the set of strong arcs. Let $K = |A^*|$.  Thus $K < m$.  Note that $\x ^*$ is obtained as the average of $K$ basic flows.  Assuming that all capacities are integral, for each arc $(i, j) \in A^*$, $x^*_{ij} \le q_{ij} - \frac{1}{K} < q_{ij} - \epsilon$.

Let $(\w^*, \v^*)$ be derived from $\x ^*$  as follows.  If  $x^*_{ij}< q_{ij} - \epsilon$, then $w^*_{ij} = x^*_{ij}$ and $v^*_{ij} = 0$.  If  $x^*_{ij}= q_{ij}$, then $w^*_{ij} = x^*_{ij} - \epsilon$ and $v^*_{ij} = \epsilon$.   Note that $(\w^*, \v^*)$ is feasible for the maximum cost flow problem (\ref{eqn:max_xy}).
Let $\eta^*$ be the optimal objective for the maximum circulation problem.   Then $f(\w^*, \v^*) = \eta^* - \frac{(m-K) \epsilon}{m+1} > \eta^* - \epsilon  $. 

We next establish the optimality of $\x'$ for  (\ref{eqn:max-circ}).
Note that $f(\w', \v') \ge f(\w^*, \v^*) > \eta^* - \epsilon$.  Therefore, $\c \cdot \x' > \eta^* - \epsilon$.   Since $(\w', \v')$ is a basic solution, it follows that all flows are integral multiples of $1/\epsilon$, and thus $\c \cdot \x' \ge \eta^* $, implying that $\x'$ is optimal for the maximum circulation problem.

We next show that $A' = A^*$, and thus, according to Corollary \ref{lem:strong-max}, $\x'$ is a strong maximum circulation.  Since $f(\w', \v') = \eta^* - \frac{(m-|A'|)\epsilon}{m+1} \ge 
f(\w^*, \v^*) = \eta^* - \frac{(m - |A^*|)\epsilon}{m+1}$, it follows that $|A'| \ge |A^*|$.  And because $\x'$ is a maximum circulation and $\x^*$ is a strong maximum circulation, it follows that $A' \subseteq A^*$.  We conclude that $A' = A^*$, and therefore $\x'$ is a strong maximum circulation.

Because $A' = A^*$, it follows that $f(\w^*, \v^*) = f(\w', \v')$.  Therefore, $(\w^*, \v^*)$ is optimal for (\ref{eqn:max_xy}).

We next show that $\x^*$ and $\y'$ satisfy strong complementary slackness conditions for (\ref{eqn:max-circ}).  We see why as follows.  Because $(\w^*, \v^*)$ is optimal for (\ref{eqn:max_xy}),  $(\w^*, \v^*)$ and $\y'$ satisfy complementary slackness conditions for (\ref{eqn:max_xy}).  Therefore,  the following is true: 
 
\begin{enumerate}
\item  For each $(i,j) \in A$, if $x^*_{ij} = 0$, then $w^*_{ij} = 0$, and $c_{ij} - y'_i + y'_j  \le  0$.
\item  For each $(i,j) \in A$, if $0 < x^*_{ij} < q_{ij}$, then $0 < w^*_{ij} < q_{ij} - \epsilon$, and  $c_{ij} - y'_i + y'_j  =  0$.
\item  For each $(i,j) \in A$, if $x^*_{ij} = q_{ij}$, then $v^*_{ij} = \epsilon$, and $(c_{ij} - \frac{1}{m+1}) - y'_i + y'_j  \ge  0$, which implies that $c_{ij} - y'_i + y'_j  >  0$.
 \end{enumerate} 
 
Thus, $\x^*$ and $\y'$ satisfy strong complementary slackness.

Finally, we show that $\x'$ and $\y'$ satisfy strong complementary slackness, which will complete the proof of this lemma:  Clearly, $\x'$ and $\y'$ satisfy the first and third conditions of strong complementary slackness.   As for the second condition, 

$$ \text{If } 0 < x'_{ij} < q_{ij}, \text{ then } 0 < x^*_{ij} < q_{ij} - \epsilon, \text{ and } c_{ij} - y'_i + y'_j  =  0.$$
\end{proof}

\section{Dual Optimal Rankings and the Hochbaum-Levin Model}\label{sec:HL}

We now discuss partial orders that are consistent with the dual linear program of the maximum circulation problem. We begin by introducing the separation-deviation model of \cite{hochbaum_methodologies_2006}, which links the Kemeny model to maximum cycle-removal methods, and we show that the convex relaxation of Kemeny's model is the dual of the maximum circulation problem (but not the strong maximum circulation problem).

The Hochbaum and Levin model, henceforth called the HL model, determines a score for each alternative so that the vector of scores minimizes a loss function. The decision variables for the HL model are the components of the vector $\y$, where $y(i)$ is the score of alternative $i$; the higher the score, the more preferred the alternative. The variables
$y(i)$ can take (continuous) values in an interval, or they can be restricted to take integer values in a given range. From an optimal solution $\y'$, a partial order of alternatives can be derived in which $i \succ j$ if $y'(i) > y'(j)$.

The voting system in the HL model is more general than the preference voting system analyzed to this point in the paper. As such, we refer to voters in this section as {\em reviewers}. The HL models permits each reviewer $r$ to report a numeric score for each alternative $i$, denoted $s^r_i$, where higher scores represent more preferred alternatives. They may also report pairwise preference values that express the intensity of the  preference for one alternative over another, denoted $p^r_{ij}$. For example, $p^r_{ij}$ could equal the difference in ranks among pairs of alternatives in a reviewer's ranked list.

For each reviewer and alternative, there is a loss $g_i^r(y(i) - s^r_i)$, which is a function of the {\em deviation} of the solution $\y'$ from the score that reviewer $r$ provides for alternative $i$.  For each reviewer $r$ and for each pair $i,j$ of alternatives, there is also a loss {\em separation} function $f^r_{i,j}(p^r_{ij} - (y'(i)- y'(j)))$, which is a function of the solution's pairwise difference in scores as compared to the intensity of the preference of reviewer $r$ for this pair $i,j$. Both the separation and deviation functions take the value $0$ for zero valued arguments. That is, if the reviewer's score or pairwise preference value is identical to the solution's, then the loss penalty is zero.

We let $V$ be the set of alternatives, and $A$ the set of pairs $(i,j)$ where at least one reviewer expresses a preference of $i$ over $j$. Let $R_{ij}$ be the set of reviewers that express a preference of $i$ over $j$. Let $R_i$ be the set of reviewers that assigned a score to alternative $i$. The HL model is to minimize the aggregate loss function which equals the total deviation and separation loss for all voters:
$$({\rm HL})\ \ \ \min Loss(\y)=\sum _{i\in V} \sum _{r\in R_i}g_i^r(y(i)-s_i^r) + \sum _{(i,j)\in A} \sum _{r\in R_{ij}} f_{ij}^r(y(i)-y(j)-p^r_{ij}).$$

Hochbaum and Levin showed how to solve the optimization problem (HL) in polynomial time when the functions $g_i^r()$ and $f_{ij}^r()$  are convex. They showed how to solve this convex problem in polynomial time for integer variables (and also for continuous variables). This problem was shown to be a special case of the {\em convex dual of minimum cost network flow}, which is solved by the algorithm of \cite{AHO03} in $O(nm\log {\frac{n^2}{m}}\log nU )$, where $n$ is the number of alternatives $|V|$, $m$ is the number of different $f_{ij}$ functions, and $U$ is the range of the variables. For example, if the components of $\y$ are constrained to fall in the range $[-n,n]$, then $U = 2n$.  When the functions are non-convex, then the problem is NP-hard, see, e.g., \cite{GT2000}.  


We next observe that Kemeny's model can be represented as a special case of the HL model applied to a preference voting system in which $R_i = \emptyset$ for all $i\in V$, and where $p^r_{ij} = 1$ for all $r \in R_{ij} $. That is, whenever reviewer $r$ prefers alternative $i$ to alternative $j$, it is associated with an intensity of $1$ unit.

Specifically, let $\delta^+(z)$ denote the function that takes the value $1$ when $z >0$ and zero otherwise. With this notation, Kemeny's model can be expressed as a special case of the HL model with a ``0-1'' loss function:

$$(\text{KEM}) \,\,\,\,\, \min  \sum _{(i,j)\in A}  q_{ij}\delta^+(y(j)-y(i)+1).$$

This objective function is non-convex, and the problem is NP-hard, as is Kemeny's model. Taking a convex relaxation of the objective of KEM via a hinge loss function, one obtains the resulting optimization problem: 

$$(\text{Relax-KEM}) \,\,\,\,\, \min  \sum _{(i,j)\in A}  q_{ij}\max\{(y(j)-y(i)+1), 0\}.$$
 
 \cite{gupte_finding_2011} created a model equivalent to Relax-KEM. where they referred to the loss function as the ``agony'' loss function.  They also showed that Relax-KEM is the dual of the problem that we refer to as the maximum circulation problem. This result is a special case of the fact that convex HL and the convex separation-deviation model is a convex cost dual network flow problem 
 \citep{hochbaum_methodologies_2006,hochbaum2001MRF}.  We provide their result as well as a more concise proof next.

 \begin{lemma}
 \label{lem:DualKEM}
 (Relax-KEM) is the equivalent to the dual of the maximum circulation problem (\ref{eqn:max-circ}).    
\end{lemma}
\begin{proof}   
The dual to the maximum circulation problem is LP (\ref{eq:dual-program}). Recall that any optimal solution $(\y, \z)$ to \ref{eq:dual-program} satisfies the following:   $z_{ij}=\max \{ y_j-y_i +1, 0\}$. Therefore, if $(\y, \z)$ is optimal for (\ref{eq:dual-program}), then
$y(i) = y_i$ minimizes (Relax-KEM).
 \end{proof}

Consequently, the maximum circulation problem (\ref{eqn:max-circ}) can be also viewed as the dual of a convex relaxation of the ``0-1'' loss minimization formulation of Kemeny's method's (KEM).

\section{Minimum Maximal Circulations} \label{sec:minmax}

We next explore the possibility of removing a maximal circulation that has a minimum number of votes. We refer to this problem as the {\em minimum maximal circulation problem}. 
Equivalently, it can be expressed as the problem of finding a circulation vector $\x$ that minimizes  $\1 \cdot \x$ such that $A(\q-\x)$ is acyclic.  The problem for Kemeny's method is a relaxation of the minimum maximal circulation problem in which $\x$ is not required to be a circulation.   That is, the goal of Kemeny's method is to identify a vector $\x$ that minimizes  $\1 \x$ such that $A(\q-\x)$ is acyclic.

While the minimum maximal circulation problem does have the desirable property of having as many votes as possible included in the acyclic graph, it shares the following two drawbacks with Kemeny's method. 
\begin{enumerate}
    \item {\bf Non-uniqueness}: There may be more than one minimum maximal circulation. Further, these minimum maximal circulations may induce conflicting partial orders. That is, there may be minimum maximal circulations $\x'$ and $\x''$ and a pair of alternatives $i$ and $j$ such that $(i,j) \in A^{TC}[V,\q - \x'] $ and $(j, i)\in A^{TC}[V,\q - \x'']$, as shown in Subsection \ref{sec:minimax-equitable}. 
    \item {\bf NP-hardness}: Finding a minimum maximal circulation is NP-hard, as shown in Subsection \ref{sec:minmaxNPhard}.
\end{enumerate}

\subsection{Minimum Maximal Circulations Can Induce Conflicting Partial Orders}\label{sec:minimax-equitable}

In this subsection, we show that there may be two different minimum maximal circulations that lead to conflicting partial orders. That is, alternatives receive more or less favorable treatment depending on which set of votes is eliminated.

Consider the vote graph $G = (V, A)$, illustrated in the top panel of \autoref{fig:minmax}.  Here $V = \{1, 2, ..., 8\}$, and $A = \{(1, 2), (2, 3), (2,5), $ $(3, 4), (4, 1), (4, 7), (5, 6), (6, 1), (6, 5), (7, 8), (8, 3), (8, 7)\}$ where each arc in $A$ has a capacity of 1.

\begin{figure}[t]
\[G[V,\q]=
\vcenter{\hbox{
\begin{tikzpicture}[main/.style = {draw, circle}, node distance=1cm, scale=0.8, transform shape]  
\node[main] (1) {$1$}; 
\node[main] (2) [right= 0.9cm of 1] {$2$};
\node[main] (3)[below right=.9cm and 0.9cm of 2] {$3$}; 
\node[main] (4)[below= 0.9cm of 3] {$4$}; 
\node[main] (5)[below left=.9cm and 0.9cm of 4] {$5$};
\node[main] (6)[left=.9cm of 5] {$6$};
\node[main] (8)[below left=.9cm and 0.9cm of 1] {$8$};
\node[main] (7)[below= 0.9cm of 8] {$7$}; 
\path[bend left=25] (1) edge [->]  (2);
\path[bend left=25] (2) edge [->]  (3);
\path[bend left=25] (3) edge [->]  (4);
\path[bend left=25] (2) edge [->]  (5);
\path[bend left=25] (4) edge [->]  (1);
\path[bend left=25] (8) edge [->]  (7);
\path[bend left=25] (7) edge [->]  (8);
\path[bend left=25] (5) edge [->]  (6);
\path[bend left=25] (6) edge [->]  (5);
\path[bend left=25] (6) edge [->]  (1);
\path[bend left=25] (4) edge [->]  (7);
\path[bend left=25] (8) edge [->]  (3);
\end{tikzpicture}
}}\]

\[
G[V,\x']=
\vcenter{\hbox{
\begin{tikzpicture}[main/.style = {draw, circle}, node distance=1cm, scale=0.8, transform shape]   
\node[main] (1) {$1$}; 
\node[main] (2) [right= 0.9cm of 1] {$2$};
\node[main] (3)[below right=.9cm and 0.9cm of 2] {$3$}; 
\node[main] (4)[below= 0.9cm of 3] {$4$}; 
\node[main] (5)[below left=.9cm and 0.9cm of 4] {$5$};
\node[main] (6)[left=.9cm of 5] {$6$};
\node[main] (8)[below left=.9cm and 0.9cm of 1] {$8$};
\node[main] (7)[below= 0.9cm of 8] {$7$}; 
\path[bend left=25] (3) edge [->]  (4);
\path[bend left=25] (7) edge [->]  (8);
\path[bend left=25] (5) edge [->]  (6);
\path[bend left=25] (6) edge [->]  (5);
\path[bend left=25] (4) edge [->]  (7);
\path[bend left=25] (8) edge [->]  (3);
\end{tikzpicture}
}}
\hspace{0.25in}
G[V,\q-\x']=
\vcenter{\hbox{
\begin{tikzpicture}[main/.style = {draw, circle}, node distance=1cm, scale=0.8, transform shape]   
\node[main] (1) {$1$}; 
\node[main] (2) [right= 0.9cm of 1] {$2$};
\node[main] (3)[below right=.9cm and 0.9cm of 2] {$3$}; 
\node[main] (4)[below= 0.9cm of 3] {$4$}; 
\node[main] (5)[below left=.9cm and 0.9cm of 4] {$5$};
\node[main] (6)[left=.9cm of 5] {$6$};
\node[main] (8)[below left=.9cm and 0.9cm of 1] {$8$};
\node[main] (7)[below= 0.9cm of 8] {$7$}; 
\path[bend left=25] (1) edge [->]  (2);
\path[bend left=25] (2) edge [->]  (3);
\path[bend left=25] (2) edge [->]  (5);
\path[bend left=25] (4) edge [->]  (1);
\path[bend left=25] (8) edge [->]  (7);
\path[bend left=25] (6) edge [->]  (1);
\end{tikzpicture}
}}
\]

\[
G[V,\x'']=
\vcenter{\hbox{
\begin{tikzpicture}[main/.style = {draw, circle}, node distance=1cm, scale=0.8, transform shape]   
\node[main] (1) {$1$}; 
\node[main] (2) [right= 0.9cm of 1] {$2$};
\node[main] (3)[below right=.9cm and 0.9cm of 2] {$3$}; 
\node[main] (4)[below= 0.9cm of 3] {$4$}; 
\node[main] (5)[below left=.9cm and 0.9cm of 4] {$5$};
\node[main] (6)[left=.9cm of 5] {$6$};
\node[main] (8)[below left=.9cm and 0.9cm of 1] {$8$};
\node[main] (7)[below= 0.9cm of 8] {$7$}; 
\path[bend left=25] (1) edge [->]  (2);
\path[bend left=25] (2) edge [->]  (5);
\path[bend left=25] (8) edge [->]  (7);
\path[bend left=25] (7) edge [->]  (8);
\path[bend left=25] (5) edge [->]  (6);
\path[bend left=25] (6) edge [->]  (1);
\end{tikzpicture}
}}
\hspace{0.25in}
G[V,\q-\x'']=
\vcenter{\hbox{
\begin{tikzpicture}[main/.style = {draw, circle}, node distance=1cm, scale=0.8, transform shape]   
\node[main] (1) {$1$}; 
\node[main] (2) [right= 0.9cm of 1] {$2$};
\node[main] (3)[below right=.9cm and 0.9cm of 2] {$3$}; 
\node[main] (4)[below= 0.9cm of 3] {$4$}; 
\node[main] (5)[below left=.9cm and 0.9cm of 4] {$5$};
\node[main] (6)[left=.9cm of 5] {$6$};
\node[main] (8)[below left=.9cm and 0.9cm of 1] {$8$};
\node[main] (7)[below= 0.9cm of 8] {$7$}; 
\path[bend left=25] (2) edge [->]  (3);
\path[bend left=25] (3) edge [->]  (4);
\path[bend left=25] (4) edge [->]  (1);
\path[bend left=25] (6) edge [->]  (5);
\path[bend left=25] (4) edge [->]  (7);
\path[bend left=25] (8) edge [->]  (3);
\end{tikzpicture}
}}
\]
    \caption{Example of a vote graph where minimum maximal circulations induce conflicting partial orders.}
\label{fig:minmax}

\end{figure}
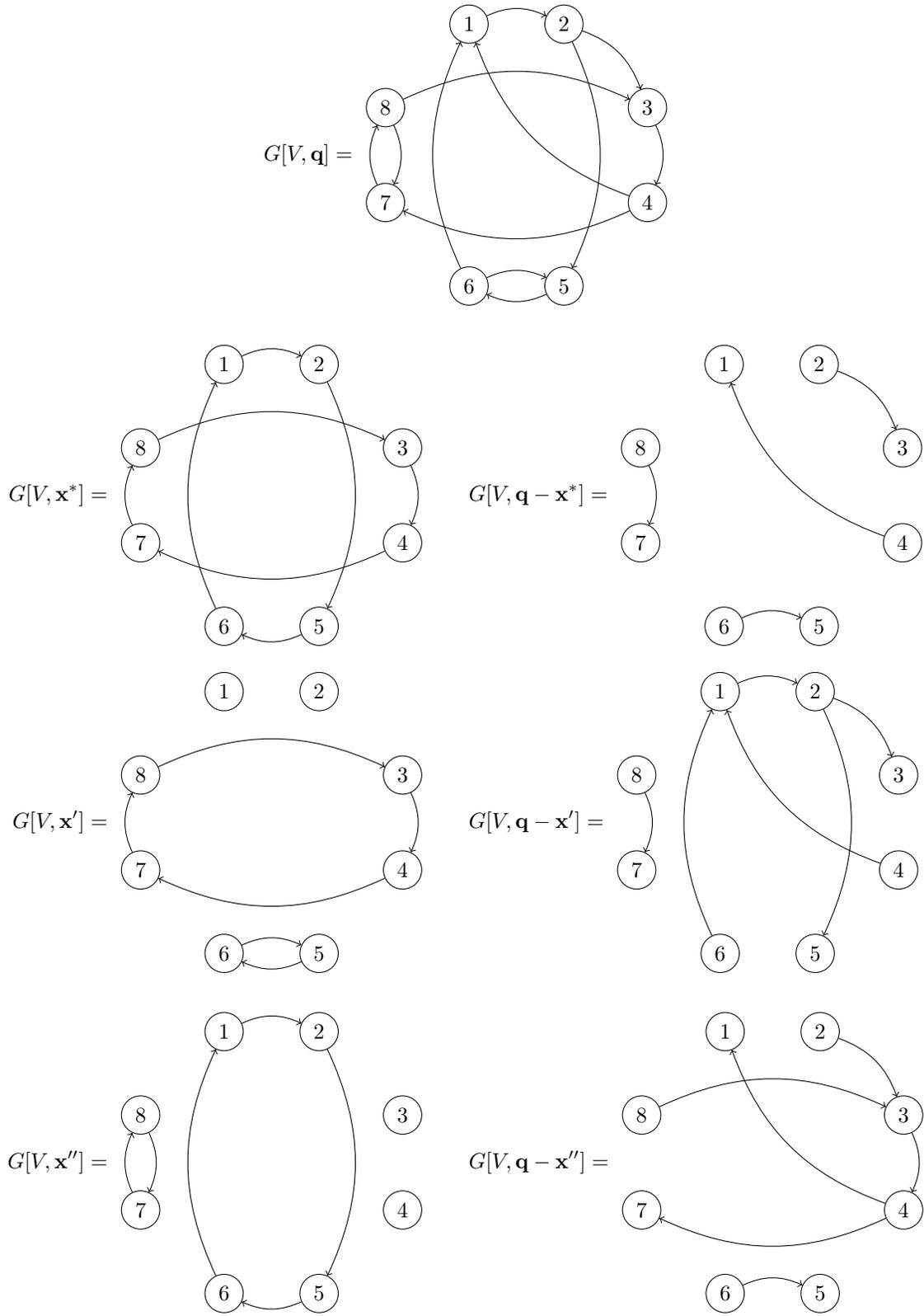

There are two minimum maximal circulations denoted $\x'$ and $\x''$ in the second and third rows of the figure, which contain 6 votes each. $\x'$ consists of directed cycles 3-4-7-8-3 and 5-6-5. $\x''$ consists of directed cycles 1-2-5-6-1 and 7-8-7. After eliminating $\x'$, there remains a directed path from node 1 to node 3.  Therefore, 1 is preferred to 3 in $A^{TC}(\q - \x')$.  After eliminating $\x''$, there remains a directed path from node 3 to node 1. Therefore, 3 is preferred to 1 in $A^{TC}(\q - \x'')$. Thus, the two partial orders conflict.

A strong maximum circulation $\x^*$ can be found by taking a convex combination of two maximum circulations.  The first is a unit flow around cycles 1-2-5-6-1 and 3-4-7-8-3. The second consists of unit flows around cycles 1-2-3-4-1, 5-6-5, and 7-8-7.  $A^{TC}(\q - \x^*)=\{(2,3),(2,5),(4,1),(4,7),(6,1),(6,5),(8,3),(8,7)\}$.  Note that in $A^{TC}(\q - \x^*)$, elements 1, 3, 5, and 7 do not dominate any other elements, and are all unrelated. 

\subsection{Computational Complexity of Finding a Minimum Maximal Circulation}\label{sec:minmaxNPhard}

We next demonstrate the NP-hardness of finding a minimum maximal circulation, which we refer to as the \emph{minmax circulation problem}.  To prove that the problem is NP-complete we formalize this problem as a {\em decision problem}.  If a decision problem is NP-complete, then the respective optimization version is said to be NP-hard.

Recall that the minimum FAS problem, which we restate here as a decision problem, is NP-complete.  As a result, finding a minimum FAS is NP-hard \citep{garyjohnson}.\\

\noindent
{\bf Minimum Feedback Arc Set} $(G,K)$\\
INSTANCE:   A directed graph $G = (V, A)$ and integer $K$.\\
QUESTION:   Is there a subset of at most $K$ arcs whose removal from $G$ leaves an acyclic graph?\\

We next state the minmax circulation problem as a decision problem and prove that it is also NP-complete:\\

\noindent
{\bf The Minmax Circulation Problem} $(G,K)$\\
INSTANCE:   A directed multi-graph $G = (V, A)$ and integer $K$.\\
QUESTION:   Is there a circulation with at most $K$ arcs whose removal from $G$ leaves an acyclic graph?\\

\begin{theorem}   The Minmax Circulation Problem is NP-complete.
\end{theorem}
\begin{proof}  
It is clearly in the class NP. We will carry out a transformation from the Feedback Arc Set Problem.  (This was also the starting point in \citep{bartholdi_voting_1989} for  establishing that a decision version of Kemeny's method was NP-complete.)  Let $G = (V, A)$ and $K$ be an instance of the Feedback Arc Set Problem.   We create a graph $G' = (V', A')$ for the Minmax Circulation Problem as follows. 

\begin{enumerate}
\item  $V' = V\cup \{s, t\} \cup V^*$, where $ V^*$ is defined as follows. For every arc $(i, j)\in A$, we replace $(i, j)$ by a path $P[i, j]$ of length $3$.  The first and last nodes of $P[i, j]$ are $i$ and $j$.  The second and third nodes are additional nodes added to $V^*$ that we label $ij_1$ and $ij_2$.
\item $A'$ is defined as follows:  For each arc $(i, j)\in A$, $A'$ includes the arcs of the path
$P[i, j] = (i, ij_1), (ij_1, ij_2), (ij_2, j)$.  In addition, $A'$ includes an arc $(s, ij_1)$ and an arc $(ij_2, t)$.  ($A'$ does not include the arc $(i,j)$.)
\item $A'$ also includes $K$ copies of arc $(t, s)$.  ($A'$ contains no other arc.)
\item $K' = 4K$.
\end{enumerate}

Note that we have introduced a number of cycles of length four in the graph $G'$.  For each $(i, j)\in A$, let $C[i, j] = s$-$ij_1$-$ij_2$-$t$-$s$.  Thus $C[i, j]$ has four arcs.   There is no cycle in $A'$ with fewer than four arcs.

We claim that there is a feasible solution for the minmax circulation problem $(G', K')$ if and only if there is a feasible solution for the feedback arc set problem instance $(G, K)$.
Suppose first that there is a subset $S \subseteq A$ of $K$ arcs whose deletion from $G$ results in an acyclic graph.
Let $S'\subseteq A'$ be obtained as follows:   For each $(i, j) \in A$, if $(i, j) \in S$, then $S'$ contains the cycle $C[i, j]$.  Accordingly, $S'$ contains $K$ cycles, each with four arcs.  The graph with arc set $A'\setminus S'$ is acyclic.

Suppose conversely that $S' \subseteq A'$ has at most $K'$ arcs and suppose that $A'\setminus S'$ is acyclic.  $S'$ must contain all $K$ copies of $(t, s)$.  Moreover, each simple cycle in $A'$ contains at least $4$ arcs.  Since $K' = 4K$, each of the cycles in $S'$ contain exactly four arcs.
The only cycles in $A'$ with four arcs are the sets $C[i, j]$ for $(i, j) \in A$.  Let $S = \{(i, j) : C[i, j] \subset S'\}$.   Then $A\setminus S$ is acyclic.		
\end{proof}

\cite{guruswami_beating_2011} proved that if the Unique Games Conjecture is true, then it is hard to approximate the minimum feedback arc set to within any constant factor in polynomial time.  
Because of the nature of our reduction (factor $4$, constant, used in the reduction), we have the following corollary:

\begin{corollary}   If the Unique Games Conjecture is true, then it is hard to approximate the minimum maximal circulation to within any constant factor in polynomial time.  
\end{corollary}

\section{Conclusion}

It is self evident that a cycle of preference votes, each requiring that one alternative ranks higher than the next, cannot conform to any partial order, or consistent ranking. Such a cycle in a vote graph represents a tie, or contradictory votes for the alternatives on the cycle. The fundamental problem of preference aggregation is determining what information in the vote graph should be disregarded in order to obtain a consensus partial order in this setting \citep[see, e.g.,][]{schwartz_cycles_2018}. 

The approach proposed here is to remove unions of cycle of votes, i.e., maximal circulations in the vote graph, and infer aggregate preferences from the non-conflicting remainder. We begin by examining methods that are based on the removal of a maximum circulation. A major contribution of our paper is the introduction of the {\em strong} maximum circulation that guarantees the {\em uniqueness} of the induced {\em strong partial order}. If there is {\em any} maximum circulation that keeps at least one vote for a certain preference, then so does the strong maximum circulation. A second major contribution is that we show how to obtain a strong maximum circulation by solving a single minimum cost flow problem.




We then compare the maximum circulation-removal approach with existing optimization-based approaches, focusing on the well-known Kemeny method. Kemeny's method does not return a unique solution and may result in contradictory rankings among the different solutions. Further, there is no polynomial time algorithm to produce an optimal solution for Kemeny's method (unless P = NP). The strong maximum circulation problem, in contrast, is solved in strongly polynomial time using combinatorial network flow algorithms. However, we also demonstrate a close relationship between the two methods, as the dual of the maximum circulation problem can be viewed as a convex relaxation of Kemeny's model. We also explore the possibility of removing minimum maximal circulations. We show that this approach inherits the drawbacks associated with Kemeny's model related to the multiplicity of inconsistent solutions and computational complexity.

We conclude by discussing extensions of the approach to preference aggregation introduced here. We begin by noting that there are alternative efficient maximal cycle removal approaches that remove fewer votes than a strong maximum circulation does that also inherit many of the desirable properties of the strong maximum circulation algorithm. For example, it is straightforward to first eliminate all cycles of length 2, and then find a strong maximum circulation on the resulting graph. Doing so will frequently eliminate fewer votes than the strong maximum circulation and will also guarantee that any alternative that beats every other alternative in pairwise comparison is the unique winner. Exploring other efficient cycle removal approaches that also induce unique partial orders---and exploring the properties of these induced aggregate preferences---is a worthwhile avenue for future investigation. Along these lines, we look forward to evaluating the relative performance of these cycle removal approaches for preference aggregation as compared with alternative optimization-based methods \citep[e.g.,][]{hochbaum_methodologies_2006, jiang_statistical_2011, negahban2012iterative}, in terms of their ability to recover a true partial order in settings with noisy preference data and as part of ensemble methods that combine the output of multiple ranking algorithms \citep[e.g.,][]{feng_ranking_2024}.

\section*{Acknowledgement}
Dorit Hochbaum's research is supported in part by AI institute NSF award 2112533.

\bibliographystyle{aea}
\bibliography{library.bib}

\appendix

\end{document}